\documentclass{birkjour}
 \theoremstyle{definition}
  \theoremstyle{remark}
   \numberwithin{equation}{section}
\usepackage{mathptmx}
\usepackage{latexsym}
\usepackage{verbatim}
\usepackage{tikz}
\usepackage{amsmath}
\usepackage{graphicx}
\usepackage{amsfonts}
\usepackage[utf8]{inputenc}
\usepackage{pdfpages}
\newtheorem{teo}{Theorem}
\newtheorem{lema}{Lemma}

\begin{document}

\title[]{Uniqueness of the Gibbs state of the BEG model in the disordered region of parameters}

\author[Paulo C. Lima]{Paulo C. Lima}

\address{Departamento de Matemática
\\
Universidade Federal de Minas Gerais\\
Belo Horizonte - MG\\
Brazil}

\email{lima@mat.ufmg.br}

\maketitle

\begin{abstract} We show that the $d$-dimen\-sion\-al Blume-Emery-Griffiths model has a unique Gibbs state, for all temperature,  in some portion of   disordered region of parameters, ruling out the possibility of a reentrant behavior in the same.
\end{abstract}

\section{Introduction}

The Blume-Emery-Griffiths (BEG) model is a spin-one system, introduced in the 1970s in the context of superfluidity and phase transition of  $\,{}^3He-{}^4He$ mixtures \cite{bib:rbeg} and since then  it has attracted a lot of attention and has been extended to other applications such as ternary fluids \cite{bib:muk,bib:fur}, phase transitions in $UO_2$ \cite{bib:grif} and $DyVO_4$ \cite{bib:blume},  phase changes in microemulsion \cite{bib:ss} and solid-liquid-gas systems \cite{bib:lss}.

The formal Hamiltonian of  $BEG$ model with zero magnetic field has the following form:
\begin{align*}{\mathcal H}(\sigma)=-\sum_{\langle i, j\rangle}(\sigma_i\sigma_j+y\sigma_i^2\sigma_j^2+x(\sigma_i^2+\sigma_j^2)),\end{align*}
where $\langle i, j\rangle$ is an unordered pair of  nearest neighbors in $\mathbb{Z}^d$, $\sigma_i\in \{-1,0,+1\}$ and $x,y\in\mathbb{R}$.

To understand the low temperature properties of the model it is important to know its low energy configurations and this is done in \cite{bib:pgm}, where the $xy$-plane is decomposed into three  regions (according to the lowest spin pair energies),  namely,
\begin{eqnarray*}
{\mathcal F}&=&\{(x,y):1+2x+y>0 \mbox{ and } 1+x+y>0\}\\
{\mathcal D}&=&\{(x,y):1+2x+y<0 \mbox{ and } x<0\}\\
{\mathcal A}&=&\{(x,y):1+x+y<0 \mbox{ and } x>0\},
\end{eqnarray*}
 called ferromagnetic, disordered and antiquadrupolar. In these regions the spin pairs with lowest energies  are $\{++,--\}$, $\{00\}$ and $\{0+,0-\}$, respectively.  In particular, for $(x,y)\in {\mathcal D}$ the constant configuration $\omega_i=0$, for all $i$, is the only ground state. For $(x,y)\in{\mathcal F}$ there are two ground states, namely, the constant configurations $\omega_i=+1$, for all $i$, and $\omega_i=-1$, for all $i$, respectively. For $(x,y)\in {\mathcal A}$ the model has infinitely many ground states, namely,  $\omega_i=0$ for $i\in L_e$, where $L_e$ is the even sublattice of $\mathbb{Z}^d$, and $\omega_i=\pm 1$ for $i\in L_o$, where $L_o$ is the odd sublattice of $\mathbb{Z}^d$, as well as $\omega_i=0$ for $i\in L_o$ and $\omega_i=\pm 1$ for $i\in L_e$.

      A discussion of the low temperature properties of  BEG model is found in \cite{bib:eu2018}. Since only the regions ${\mathcal A}$ and ${\mathcal D}$ are directly related to this work, we will make some comments about them.

     For $(x,y)\in {\mathcal A}$, even though we have infinitely many ground state configurations, they split into two equivalence classes and the low temperature properties of the model can be analyzed as in the extension of the Pirogov-Sinai theory given in \cite{bib:gsuto}, where ground state configurations are replaced with equivalence classes.  In \cite{bib:pgm}, using polymer expansion and analyticity techniques, the two (staggered) pure  states of the BEG model with parameters in the region ${\mathcal A}$ are constructed. The corresponding  phase, where these two pure states  coexist, is the  \emph{antiquadrupolar phase}.

     For each $(x,y)\in {\mathcal D}$,  from the usual Pirogov-Sinai theory \cite{bib:sinai,bib:sv} and high temperature expansions \cite{bib:sv}, at low and at high temperatures, respectively,  the model has a unique Gibbs state,  the \emph{disordered phase}.
     We may wonder whether or not we have a unique Gibbs state for all temperature.  This not clear at all. In fact, the phase diagrams (g), (h) and (i) given in  Figure $1$ of  \cite{bib:berker1} corresponding to the values $y=-1.5$, $-3.0$ and $-3.5$, respectively,  indicate the presence of reentrance (namely, for each one of these values of $y$ and  for $x<0$ and close to $0$ fixed,  as we increase the temperature, we go from the disordered phase to the staggered one and again to the disordered one), although this is not confirmed by the numerical renormalization group calculations (see \cite{bib:osorio,bib:branco} and references therein).   This makes the analysis of the BEG model with parameters  $(x,y)\in {\mathcal D}$ important and motivates the present work.

 Our analysis of the uniqueness of the Gibbs state in the region ${\mathcal D}$ started in  \cite{bib:eu2018}, where its Theorem $1$ combined with the FKG inequality  implied  the  uniqueness of the Gibbs state,  for all temperatures, for parameters $(x,y)$ given in Figure $1$ of this reference. However, the use of the FKG inequality is a big restriction since for the BEG model  this inequality holds only in the region $|y|\leq 1$.

Using  the Dobrushin criterion \cite{bib:dc}, we extend the uniqueness results for all temperature of \cite{bib:eu2018} to a bigger region of parameters,  ruling out the possibility of a  reentrant behavior in this region.    As a simple consequence of our results, the Blume-Capel model, which is a special case of the BEG model when $y=0$, has a unique Gibbs state for all temperature, if $x<x_c(d)$, where $x_c(2)\approx -3.69658$ and $x_c(3)\approx -3.77794$.

 This article is organized as follows: in Section \ref{secnot} we introduce some notation, we give the Dobrushin uniqueness criterion and we state our main results which are Theorems  \ref{teo0} and  \ref{teo1}.  In Section \ref{sec2} we  prove Lemma \ref{lem17.x} which provides  upper bounds  on the total variation distances, which will be used in our estimates in Lemmas \ref{lemx1} and  \ref{lemdob16.1}.  In Section \ref{secunifbc} we prove  Lemmas \ref{lemx1} and  \ref{lemdob16.1}.  Finally, in Section \ref{seccr}, we make some concluding remarks.

\section{Notation and main result} \label{secnot}

Given a finite  $\Lambda\subset\mathbb{Z}^d$,  let $\Omega_\Lambda=\{-1,0,1\}^\Lambda$.  The weight for the Boltzmann-Gibbs distribution for $\Lambda$  with external configuration $\sigma_{\Lambda^c}\in\{-1,0,1\}^{\Lambda^c}$, $\pi^{\Lambda}_{x,y,\beta}(.|\sigma_{\Lambda^c})$,  is defined  for each $\sigma_\Lambda\in\Omega_\Lambda$,  as
\begin{eqnarray}\label{eqspecif}\pi^{\Lambda}_{x,y,\beta}(\sigma_\Lambda|\sigma_{\Lambda^c})=\frac{e^{-\beta {\mathcal H}^{\Lambda}_{x,y,\beta}(\sigma_\Lambda|\sigma_{\Lambda^c})}}{\sum_{\xi_\Lambda\in \Omega^\Lambda}e^{-\beta {\mathcal H}^{\Lambda}_{x,y,\beta}(\xi_\Lambda|\sigma_{\Lambda^c})}},\end{eqnarray}
where $\beta\geq 0$ is the inverse temperature and $${\mathcal H}^{\Lambda}_{x,y,\beta}(\sigma_\Lambda|\sigma_{\Lambda^c})=-
\sum_{\langle i,j\rangle \,:\, \{i,j\}\cap \Lambda\neq \emptyset}(\sigma_i\sigma_j+y\sigma_i^2\sigma_j^2+x(\sigma_i^2+\sigma_j^2)).$$

A probability measure $\mu$ on the configuration space $\{-1,0,1\}^{\mathbb{Z}^d}$ is said to be an infinite-volume Gibbs measure (or just Gibbs measure or Gibbs state) for the formal Hamiltonian  ${\mathcal H}$ if, for each finite subset $\Lambda\subset \mathbb{Z}^d$, the conditional probability distribution $\mu(.|\sigma_{\Lambda^c})$ equals $\pi^{\Lambda}(.|\sigma_{\Lambda^c})$, where $\pi^{\Lambda}(.|\sigma_{\Lambda^c})$ is given by (\ref{eqspecif}).   Since our space of states is $\{-1,0,1\}$,  which is finite, by a compactness argument, we can show there exists at least one Gibbs measure $\mu$ for the formal Hamiltonian  ${\mathcal H}$, see for instance, Theorem $1.2$ of \cite{bib:sinai}.

If $\sigma$ and $\tilde{\sigma}$ are any two boundary conditions, the  total variation distance between $\pi^{\{i\}}_{x,y,\beta}(.|\sigma)$ and $\pi^{\{i\}}_{x,y,\beta}(.| \tilde{\sigma})$ is defined (see, for instance, equations (2.3) and (8.3)  of \cite{bib:dc} and \cite{bib:geor}, respectively) as
$$d(\pi^{\{i\}}_{x,y,\beta}(.|\sigma),\pi^{\{i\}}_{x,y,\beta}(.|\tilde{\sigma}))=
\frac{1}{2}\sum_{\xi\in\{-1,0,1\}}|\pi^{\{i\}}_{x,y,\beta}(\xi|\sigma)-\pi^{\{i\}}_{x,y,\beta}(\xi|\tilde{\sigma})|.$$

The Dobrushin uniqueness criterion, see Theorems $2$ and $8.7$ of \cite{bib:dc} and  \cite{bib:geor}, respectively,  establishes that there is at most one  Gibbs measure, $\mu_{x,y,\beta}$, for the formal Hamiltonian ${\mathcal H}$, if
 \begin{eqnarray}\label{eqdobc16.1}\sup_{i\in\mathbb{Z}^d}\sum_{j\in\mathbb{Z}^d\backslash \{i\}}\max_{\sigma\equiv \tilde{\sigma} \mbox{ off $j$}}d(\pi^{\{i\}}_{x,y,\beta}(.|\sigma),\pi^{{\{i\}}}_{x,y,\beta}(.|\tilde{\sigma}))<1,\end{eqnarray}
where  $\sigma\equiv \tilde{\sigma} \mbox{ off $j$}$, means that the configurations $\sigma$ and  $\tilde{\sigma}$ are different only at the site $j$.  Therefore, for the BEG model  the validity of (\ref{eqdobc16.1}) implies the existence of exactly one Gibbs measure.

 Since for the BEG model the interactions have range one, we may assume that  $j$ is  one of the $2d$ nearest neighbors of $i$; otherwise, $\sigma$ and $\tilde{\sigma}$ will coincide on the the boundary of $\{i\}$ and  $d(\pi^{\{0\}}_{x,y,\beta}(.|\sigma),\pi^{\{0\}}_{x,y,\beta}(.|\tilde{\sigma}))=0$. Moreover, from the Markov property, we need the specification of $\sigma$ (and so $\tilde{\sigma}$, since they differ only at $j$) only at the $2d$ nearest neighbors of  $i$.  Clearly, the left hand side of (\ref{eqdobc16.1}) does not depend on $i$, which we may assume to be the origin $0$. Also it does not depend on $j$ and we may take it as  $(1,0, \ldots,0)$, which we will denote just by $1$ (for the sake of notation, the other $2d-1$ nearest neighbors of the origin will be denoted just by $2$, \ldots, $2d$).  Therefore, the condition (\ref{eqdobc16.1}) can be replaced with
 \begin{eqnarray}\label{eqdobc16.2} \max_{\sigma\equiv \sigma' \mbox{ off $1$}}d(\pi^{\{0\}}_{x,y,\beta}(.|\sigma),\pi^{\{0\}}_{x,y,\beta}(.|\tilde{\sigma})<1/2d. \end{eqnarray}

Before stating our main results, we need few definitions.    By $\sigma\equiv \tilde{\sigma} \mbox{ off $1$}$ we will mean  $\sigma,\tilde{\sigma}\in\{-1,0,1\}^{2d}$  with $\tilde{\sigma}_1\neq \sigma_1$ and $\tilde{\sigma}_i=\sigma_i$, for $i=2,3,\ldots, 2d$.   Let
\begin{eqnarray}\label{defa}
A&=&\{(x,y):x+y+1<0, \,x<0, \, y\geq 1\}\\
B&=&\{(x,y):x+y+1<0,\, x<0, \, -1<y<1\}\label{defb}\\
C&=&\{(x,y):x+y+1<0,\, x<0, \, y\leq -1\}\label{defc}
\end{eqnarray}
and for any $t>0$ and $(x,y)\in {\mathcal U}\equiv A\cup B\cup C\subset {\mathcal D}$, define
\begin{eqnarray}\label{eqmax}r(t)&=&\frac{4}{1+t}\, \frac{1}{(1+1/t)^t}\\
\label{eqjul9.1}
a(d,x,y)&=&\left\{\begin{array}{ll}
2d|x+y+1| & \mbox{if $(x,y)\in A\cup B$}\\
2d|x|&\mbox{if $(x,y)\in  C$}\end{array} \right.\\ \label{eqjul9.2}
b(y)&=&\left\{\begin{array}{ll}
y+1 & \mbox{if $(x,y)\in A$}\\
2&\mbox{if $(x,y)\in B$}\\
|y|+1&\mbox{if $(x,y)\in C$}
\end{array} \right..
\end{eqnarray}

The next result is the most technical one, its proof will  be given in Section \ref{secunifbc}.
\begin{teo} \label{teo0} If $(x,y)\in A\cup B\cup C$, where $A,B$ and $C$ are defined by (\ref{defa}), (\ref{defb}) and (\ref{defc}), respectively, and $\sigma\equiv \tilde{\sigma} \mbox{ off $1$}$,  then

\begin{align*}
d(\pi^{\{0\}}_{x,y,\beta}(.|\sigma),\pi^{\{0\}}_{x,y,\beta}(.|\tilde{\sigma}))&\leq
 4e^{-\beta a(d,x,y)}(1-e^{-\beta b(y)}),
\end{align*}
where $a(d,x,y)$ and $b(y)$ are given by (\ref{eqjul9.1}) and (\ref{eqjul9.2}), respectively.
\end{teo}

 Notice that the upper bounds for $d(\pi^{\{0\}}_{x,y,\beta}(.|\sigma),\pi^{\{0\}}_{x,y,\beta}(.|\tilde{\sigma}))$ appearing in  Theorem \ref{teo0}  are  of the following form  $w(a,b,\beta)=4e^{-a \beta}(1-e^{-b \beta}),$ where $a,b>0$, are given by (\ref{eqjul9.1}) and (\ref{eqjul9.2}), respectively. For $a$ and $b$ fixed,  $w(a,b,\beta)$ has a global maximum at $\beta_c=\beta_c(a,b)$, where $e^{-\beta_c b}=\frac{a}{a+b}$ and, at this critical value, we have   $w(a,b,\beta_c(a,b))=r(a/b),$ where  $r(t)$ is  given by (\ref{eqmax}), and so, for any $\beta$, we have
$w(a, b,\beta)\leq w(a,b,\beta_c(a,b))= r(a/b).$
This together with condition (\ref{eqdobc16.2}) imply that,  for all $\beta$, we have  $d(\pi^{\{0\}}_{x,y,\beta}(.|\sigma),\pi^{\{0\}}_{x,y,\beta}(.|\tilde{\sigma}))<r(a(d,x,y)/b(y)).$ In particular, if  $(x,y)\in{\mathcal U}_{Dob}^{d}\equiv  \{(x,y)\in {\mathcal U}:r(a(d,x,y)/b(y))<1/(2d)\},$ then  there is exactly one Gibbs measure for all temperature.

 In order to get the region  ${\mathcal U}_{Dob}^{d}$  we  need some numerical calculations; however, they are quite simple. In fact, the function $r(t)$ is decreasing in $t$ and, for each $d$ and $y$ fixed, the function $a(d,x,y)/b(y)$ is increasing  in $|x|$, and so, the function $r(a(d,x,y)/b(y))$ is decreasing in $|x|$, for $d$ and $y$ fixed. Therefore,  the boundary of ${\mathcal U}_{Dob}^{d}$ is the curve $x=x(d,y)$ which is the solution of $r(a(d,x,y)/b(y))=1/(2d)$. In order to find $x(d,y)$, the only numerical calculation we need is to find the solution, $t_d$, of the equation
\begin{eqnarray}\label{eqdobnumz}r(t)=1/(2d).\end{eqnarray}  For two and three dimensions we have the following  numerical values: $t_2\approx 5.39315$ and $t_3\approx 8.33383$.  Once we have $t_d$, the curve $x(d,y)$ is obtained from the relation
$t_d=a(d,x,y)/b(y)$, namely, $x(d,y)$ is the polygonal curve given by
\begin{align}\label{curveuniq}
x(d,y)&=\left\{\begin{array}{ll} -\,\frac{(t_d+2d)}{2d}(y+1)& \mbox{if $y\geq 1$}\\
-\,\frac{d(y+1)+t_d}{d}& \mbox{if $|y|<1$}\\
-\,\frac{t_d}{2d}(|y|+1)& \mbox{if $y\leq -1$}\end{array} \right..
\end{align}
As a consequence we have the following result.

\begin{teo}\label{teo1} For any $d$ and $y$, if $x<x(d,y)$, where $x(d,y)$ is given by (\ref{curveuniq}), then  there is a unique Gibbs state for all temperature.
\end{teo}

As a consequence of the above theorem, making $y=0$ in (\ref{curveuniq}), we conclude that for the Blume-Capel model there is a unique Gibbs state, for all temperature, if $x<-\,\frac{d+t_d}{d}\equiv x_c(d)$, where  $x_c(2)\approx -3.69658$ and $x_c(3)\approx -3.77794$.

\section{An upper bound  expression  for   $d(\pi^{\{0\}}_{x,y,\beta}(.|\sigma),\pi^{\{0\}}_{x,y,\beta}(.|\tilde{\sigma}))$}\label{sec2}

Our goal in this section is to prove Lemma \ref{lem17.x}, which gives an upper bound on $d(\pi^{\{0\}}_{x,y,\beta}(.|\sigma),\pi^{\{0\}}_{x,y,\beta}(.|\tilde{\sigma}))$ for a fixed pair of boundary conditions $\sigma,\tilde{\sigma}$ such that  $\sigma\equiv \tilde{\sigma} \mbox{ off $1$}$.  Since we need to compute quantities like $|\pi^{\{0\}}_{x,y,\beta}(\xi|\sigma)-\pi^{\{0\}}_{x,y,\beta}(\xi|\tilde{\sigma})|$, and these quantities are not affected by the interchanging of $\sigma$ and $\tilde{\sigma}$, our choice of  $\tilde{\sigma}_1$ will be such that $|\tilde{\sigma}_1|\geq |{\sigma}_1|$;  if $|\tilde{\sigma}_1|=|\sigma_1|$, then we will take  $\tilde{\sigma}_1=1$ and ${\sigma}_1=-1$. Moreover, for the sake of notation, let $\sigma^2=\sum_{i=1}^{2d}\sigma_i^2$.

\begin{lema}\label{lem17.x} For any two configurations $\sigma\equiv \tilde{\sigma} \mbox{ off $1$}$, we have
\begin{eqnarray}\label{eqdob15.2}
d(\pi^{\{0\}}_{x,y,\beta}(.|\sigma),\pi^{\{0\}}_{x,y,\beta}(.|\tilde{\sigma}))&\leq & \sum_{s=\pm 1}|\theta_s(\sigma,\tilde{\sigma}_1)|+|\psi(\sigma,\tilde{\sigma}_1)|,
\end{eqnarray}
 where \begin{eqnarray}\label{eqdob15.3}
\theta_s(\sigma,\tilde{\sigma}_1)&=&e^{\beta (2d x +y\sigma^2)} (e^{\beta y(\tilde{\sigma}_1^2-\sigma_1^2)}e^{\beta s (\tilde{\sigma}_1-\sigma_1)}-1)e^{\beta s (\sigma_1+\sigma_2+\dots+\sigma_{2d})}\\
\label{eqdob15.4}\psi(\sigma,\tilde{\sigma}_1)&=&2 e^{\beta (4d x +2y\sigma^2)}  e^{\beta y(\tilde{\sigma}_1^2-\sigma_1^2)}\sinh(\beta(\tilde{\sigma}_1-\sigma_1)).\end{eqnarray}
\end{lema}

 \begin{proof} For each $\xi\in\{-1,0,1\}$ and $\sigma\in\{-1,0,1\}^{2d}$, by definition,
$$\pi^{\{0\}}_{x,y,\beta}(\xi|\sigma)=\frac{f(\xi,\sigma)}{\sum_{\eta=0,\pm 1}f(\eta,\sigma)},$$
where  $f(\xi,\sigma)= e^{\beta x\sigma^2} e^{\beta \xi^2(2d x +y\sigma^2)}e^{\beta \xi(\sigma_1+\sigma_2+\ldots+\sigma_{2d})}.$  Therefore,  $$\pi^{\{0\}}_{x,y,\beta}(\xi|\sigma)=\frac{h(\xi,\sigma)}{\sum_{\eta=0,\pm 1}h(\eta,\sigma)},$$
  where $h(\xi,\sigma)= e^{\beta \xi^2(2d x +y\sigma^2)}e^{\beta \xi(\sigma_1+\sigma_2+\ldots+\sigma_{2d})}$.

 In order to express relations involving $\tilde{\sigma}$ in terms of  $\sigma$, notice  that we can write  $h(\xi,\tilde{\sigma})= g(\xi,\tilde{\sigma}_1,\sigma_1)h(\xi,\sigma),$  where
$g(\xi,\tilde{\sigma}_1,\sigma_1)= e^{\beta y\xi^2(\tilde{\sigma}_1^2-\sigma_1^2)}e^{\beta\xi(\tilde{\sigma}_1-\sigma_1)},$
 in particular,
\begin{align*}\pi^{\{0\}}_{x,y,\beta}(\xi|\tilde{\sigma})=
\frac{g(\xi,\tilde{\sigma}_1,\sigma_1) h(\xi,{\sigma})}{\sum_{\eta=0,\pm 1}g(\eta, \tilde{\sigma}_1,\sigma_1)h(\eta,{\sigma})}\end{align*}
and so,
\begin{eqnarray}\label{eqdob15.6}
\mbox{}&\mbox{} &\pi^{\{0\}}_{x,y,\beta}(\xi|\sigma)-\pi^{\{0\}}_{x,y,\beta}(\xi|\tilde{\sigma})\nonumber\\
\mbox{}&=& h(\xi,\sigma)\,\,\frac{\sum_{\eta=0,\pm1}(g(\eta,\tilde{\sigma}_1,\sigma_1)-g(\xi,\tilde{\sigma}_1,\sigma_1))h(\eta,\sigma)}{\sum_{\eta=0,\pm 1}h(\eta,{\sigma})\,\,\,\sum_{\eta=0,\pm 1}g(\eta, \tilde{\sigma}_1,\sigma_1)h(\eta,{\sigma})}.
\end{eqnarray}
Since both sums in the denominator of (\ref{eqdob15.6}) are bounded from below by $1$, then the denominator of  (\ref{eqdob15.6}) is also bounded from below by $1$.  Therefore,
\begin{align*}
|\pi^{\{0\}}_{x,y,\beta}(\xi|\sigma)-\pi^{\{0\}}_{x,y,\beta}(\xi|\tilde{\sigma})|&\leq & h(\xi,\sigma)\,\,|\sum_{\eta=0,\pm1}(g(\eta,\tilde{\sigma}_1,\sigma_1)-g(\xi,\tilde{\sigma}_1,\sigma_1))h(\eta,\sigma)|.
\end{align*}
By straightforward calculations, we have
\begin{eqnarray}\label{eq21.1}
 \mbox{}&\mbox{} &h(\xi,\sigma)\sum_{\eta=0,\pm1}(g(\eta,\tilde{\sigma}_1,\sigma_1)-g(\xi,\tilde{\sigma}_1,\sigma_1))h(\eta,\sigma)\nonumber\\
 \mbox{}&=& \left\{ \begin{array}{ll} \sum_{s=\pm1}\theta_s(\sigma,\tilde{\sigma}_1) &\mbox{ if $\xi=0$}\\
 - \theta_\xi(\sigma,\tilde{\sigma}_1)-\xi  \psi(\sigma,\tilde{\sigma}_1)&\mbox{  if $\xi\in\{-1,1\}$}\end{array}\right. \end{eqnarray}
 and this concludes the proof of the lemma.
 \end{proof}

\section{The proof of Theorem \ref{teo0}}\label{secunifbc}

The proof of Theorem \ref{teo0}  will follow immediately from  Lemmas \ref{lemx1} and  \ref{lemdob16.1}, which will be given in Subsections \ref{subsec1} and \ref{subsec2}, respectively.
From Lemma \ref{lem17.x}, upper bounds on $d(\pi^{\{0\}}_{x,y,\beta}(.|\sigma),\pi^{\{0\}}_{x,y,\beta}(.|\tilde{\sigma}))$  which are uniform  in configurations $\sigma\equiv \tilde{\sigma} \mbox{ off $1$}$ can be translated into upper bounds on $|\psi(\sigma,\tilde{\sigma}_1)|$, $|\theta_{-1}(\sigma,\tilde{\sigma}_1)|$ and $|\theta_1(\sigma,\tilde{\sigma}_1)|,$  which are uniform in such configurations. In order to accomplish this,  we will split the boundary conditions $\sigma\equiv \tilde{\sigma} \mbox{ off $1$}$ into two classes: $(i)$ those $\sigma\equiv \tilde{\sigma} \mbox{ off $1$}$  for which $|\sigma_1|= |\tilde{\sigma}_1|$ and $(ii)$ those $\sigma\equiv \tilde{\sigma} \mbox{ off $1$}$  for which  $|\sigma_1|\neq  |\tilde{\sigma}_1|$.

\subsection{ The case $\sigma\equiv \tilde{\sigma} \mbox{ off $1$}$ and  $|\sigma_1|= |\tilde{\sigma}_1|$}\label{subsec1}

\begin{lema}\label{lemx1} Let $(x,y)\in A\cup B\cup C$, where $A$, $B$ and $C$ are defined by (\ref{defa}), (\ref{defb}) and (\ref{defc}), respectively. Then for any $\sigma\equiv \tilde{\sigma} \mbox{ off $1$}$ such that $|\sigma_1|= |\tilde{\sigma}_1|,$  we have

\begin{eqnarray}\label{eqdob15.11}
d(\pi^{\{0\}}_{x,y,\beta}(.|\sigma),\pi^{\{0\}}_{x,y,\beta}(.|\tilde{\sigma}))
\leq \left\{
\begin{array}{ll}4 e^{\beta (2dx+2d(y+1)) } (1 -e^{-2\beta })&\mbox{if $(x,y)\in A\cup B$}\\
4e^{\beta(2d x+y+1) } (1 -e^{-2\beta })& \mbox{if $(x,y)\in C$.}
\end{array}
\right.
\end{eqnarray}
\end{lema}
 \begin{proof} If  $|\sigma_1|= |\tilde{\sigma}_1|$  then, as we remarked in the beginning of Section \ref{sec2}, we may assume  that $\tilde{\sigma}_1=1$ and  $\sigma_1=-1$.  For the sake of notation,  let  $\sum_{j=2}^{2d}\sigma_j^2=k$, then $k\in\{0,1,\ldots,2d-1\}$ is the number of the variables $\sigma_2, \ldots, \sigma_{2d}$ which are different from $0$.   Also, let  $\sum_{j=2}^{2d}\sigma_j=n$, which implies $n\in\{-k,...,k\}$. Therefore,    from (\ref{eqdob15.3}), (\ref{eqdob15.4}),  since $\sinh(2\beta)\leq e^{2\beta}(1-e^{-2\beta})$, we have
\begin{align}\label{eqdob15.9}
|\psi(\sigma,\tilde{\sigma}_1)|&= 2 e^{\beta (4d x +2(k+1)y)} \sinh (2\beta)\nonumber\\
\mbox{}&\leq  2 e^{\beta (4d  x +2(k+1)y+2)}  (1 -e^{-2\beta })\\
|\theta_{s}(\sigma,\tilde{\sigma}_1)|&=  e^{\beta (2d x +(k+1)y)} |e^{2\beta s }-1|e^{\beta s (-1+n)}\nonumber\\
\mbox{}&=  e^{\beta (2d x +(k+1)y+1)} (1-e^{-2\beta  })e^{\beta s n}\label{eqdez319}.
\end{align}
From (\ref{eqdez319}), since $\cosh(\beta s n)\leq e^{\beta |n|}\leq e^{\beta k}$, we have
\begin{align}\label{eqdob15.10}
\sum_{s=\pm 1}|\theta_s(\sigma,\tilde{\sigma}_1)|&\leq  2e^{\beta (2d x +(k+1)(y+1))} (1-e^{-2\beta  }).
\end{align}
Then, from (\ref{eqdob15.9}) and (\ref{eqdob15.10}), since $2dx+(k+1)y-k+1<0$ for  $(x,y)\in A\cup B\cup C$ and $1\leq k+1\leq 2d$,  we have
\begin{align*}
d(\pi^{\{0\}}_{x,y,\beta}(.|\sigma),\pi^{\{0\}}_{x,y,\beta}(.|\tilde{\sigma}))&\leq  2 e^{\beta(4d x+2(k+1)y+2) } (1 -e^{-2\beta })\\
\mbox{}&\mbox{}  +2 e^{\beta (2d x +(k+1)(y+1))} (1 -e^{-2\beta })\\
\mbox{}&=2 e^{\beta (2d x +(k+1)(y+1))} (1 -e^{-2\beta })(e^{\beta(2dx+(k+1)y-k+1)}+1)\\
\mbox{}&\leq 4 e^{\beta (2d x +(k+1)(y+1))} (1 -e^{-2\beta })\\
\mbox{}&\leq  4\,\,\max_{k=0,\ldots, 2d-1}\{e^{\beta (2d x +(k+1)(y+1))}\}\,\, (1 -e^{-2\beta })\\
\mbox{}&\leq \left\{
\begin{array}{ll}4 e^{\beta (2dx+2d(y+1)) } (1 -e^{-2\beta })&\mbox{if $(x,y)\in A\cup B$}\\
4e^{\beta(2d x+y+1) } (1 -e^{-2\beta })& \mbox{if $(x,y)\in C$}
\end{array}\right.,
\end{align*}
which proves the lemma.\end{proof}

\subsection{ The case $\sigma\equiv \tilde{\sigma} \mbox{ off $1$}$ and  $|\sigma_1|\neq  |\tilde{\sigma}_1|$}\label{subsec2}

For  $\sigma\equiv \tilde{\sigma} \mbox{ off $1$}$ such that $|\sigma_1|\neq  |\tilde{\sigma}_1|$,  as we remarked in the beginning of Section \ref{sec2}, we may assume  that $\tilde{\sigma}_1= \{-1,1\}$  and $\sigma_1=0$.  Notice that because of the  relation $\theta_i(\sigma,-\tilde{\sigma}_1)=\theta_{-i}(-\sigma,\tilde{\sigma}_1),$
 since our upper bound for $\sum_{s=\pm 1}|\theta_{s}(\sigma,1)|$ will be uniform in $\sigma$, then it will also be an upper bound for  $\sum_{s=\pm 1}|\theta_{s}(\sigma,-1)|$. Therefore, as long as  upper bounds for $\sum_{s=\pm 1}|\theta_{s}(\sigma,\tilde{\sigma}_1)|$ which are uniform in $\sigma$ are concerned, we may assume that $\tilde{\sigma}_1=1$.  On the other hand, if $\sigma_1=0$ and $\tilde{\sigma}_1=\pm 1$, since the function $\sinh(.)$ is odd, then  from (\ref{eqdob15.4}), we have
\begin{align}\label{eqdob16.2}
|\psi(\sigma,\tilde{\sigma}_1)|&= 2e^{\beta (4d x +2ky)}  e^{\beta y\,\, \tilde{\sigma}_1^2}\sinh(\beta \,\,|\tilde{\sigma}_1|)= 2e^{\beta (4d x +(2k+1)y)}\sinh \beta,
\end{align}
which does not distinguish $\tilde{\sigma}_1=1$ from $\tilde{\sigma}_1=-1$.  Therefore,  for the class of boundary conditions $\sigma\equiv \tilde{\sigma} \mbox{ off $1$}$ for which we have $|\sigma_1|\neq  |\tilde{\sigma}_1|$, we will assume $\tilde{\sigma}_1=1$ and $\sigma_1=0$.  And so, from (\ref{eqdob15.3}),  we have
\begin{align}\label{eqdob16.3}
|\theta_1(\sigma,1)|&=    e^{\beta (2d x +ky)} |1 -e^{\beta  (y+1) }| e^{\beta n}\nonumber\\
\mbox{}&= \left\{ \begin{array}{ll}
e^{\beta (2d x +ky)} (1 -e^{\beta  (y+1) }) e^{\beta n}& \mbox{if $y\leq -1$}\\
e^{\beta (2d x +(k+1)y}) (1 -e^{-\beta  (y+1) }) e^{\beta (1+n)}& \mbox{if $y>-1$}
\end{array}\right.
\end{align}
\begin{align}\label{eqdob16.4}
 |\theta_{-1}(\sigma,1)|&=  e^{\beta (2d x +ky)} |1 -e^{\beta  (y-1) }| e^{-\beta  n}\nonumber\\
\mbox{}&=\left\{ \begin{array}{ll}
e^{\beta (2d x +ky)} (1 -e^{\beta  (y-1) }) e^{-\beta n}&\mbox{if $y\leq 1$}\\
e^{\beta (2d x +(k+1)y)} (1 -e^{-\beta  (y-1) }) e^{-\beta  (1+n)}&\mbox{if $y> 1$.}\\
\end{array}\right.
\end{align}

\indent If $y\geq 1$, since $1-e^{-\beta(y-1)}\leq 1-e^{-\beta(y+1)}$, then  from (\ref{eqdob16.3}) and (\ref{eqdob16.4}), we have
\begin{align}\label{eqjuly15.2}
 \sum_{s=\pm 1}|\theta_s(\sigma,1)|&\leq 2e^{\beta(2dx+(k+1)y)}(1-e^{-\beta(y+1)})\cosh(\beta(n+1))\nonumber\\
 \mbox{}&\leq
2e^{\beta(2dx+(k+1)(y+1))}(1-e^{-\beta(y+1)}).
\end{align}
\indent If $-1<y<1$, since $1-e^{-\beta(y+1)},1-e^{\beta(y-1)}\leq 1-e^{-2\beta}$ and  $y+1\geq 0$, which implies  $(k+1)y+1\geq ky$,  then  from (\ref{eqdob16.3}) and (\ref{eqdob16.4}), we have
\begin{align}\label{eqjuly15.3}
 \sum_{s=\pm 1}|\theta_s(\sigma,1)|&\leq 2e^{\beta(2dx+(k+1)y+1)}(1-e^{-2\beta})\cosh(\beta n)\nonumber\\
 \mbox{}& \leq
2e^{\beta(2dx+(k+1)(y+1))}(1-e^{-2\beta}).
\end{align}
\indent If $y\leq -1$, since $1-e^{\beta(y+1)}\leq 1-e^{\beta(y-1)}$, then  from (\ref{eqdob16.3}) and (\ref{eqdob16.4}), we have
\begin{align}\label{eqjuly15.4}
  \sum_{s=\pm 1}|\theta_s(\sigma,1)|&\leq 2e^{\beta(2dx+ky)}(1-e^{\beta(y-1)})\cosh(\beta n)\nonumber\\
  \mbox{}&\leq
2e^{\beta(2dx+k(y+1))}(1-e^{\beta(y-1)}).
\end{align}
Therefore, from (\ref{eqjuly15.2}), (\ref{eqjuly15.3}) and (\ref{eqjuly15.4}), we have
\begin{align}\label{eqjuly8z}
\sum_{s=\pm 1}|\theta_s(\sigma,1)|&\leq
\left\{
\begin{array}{ll}
2e^{\beta(2dx+(k+1)(y+1))}(1-e^{-\beta(y+1)})& \mbox{if $y\geq 1$}\\
2e^{\beta(2dx+(k+1)(y+1))}(1-e^{-2\beta})& \mbox{if $|y|<1$}\\
2e^{\beta(2dx+k(y+1))}(1-e^{\beta(y-1)})& \mbox{if $y\leq -1$}
\end{array}
\right..
\end{align}
\indent Next we will analyse $|\psi(\sigma,1)|$. First notice that $2\sinh \beta=e^{\beta}(1-e^{-2\beta}),$ for all $\beta$. Moreover,   for $y\geq 1$, we have   $1-e^{-2\beta}\leq 1-e^{-\beta(y+1)}$ and for $y\leq -1$, we have $1-e^{-2\beta}\leq 1-e^{\beta(y-1)},$ and so, from (\ref{eqdob16.2}), we have
\begin{align}\label{eqdob16.2w1}
|\psi(\sigma,1)|&=e^{\beta (4d x +(2k+1)y+1)}(1-e^{-2\beta})\nonumber\\
\mbox{}&\leq  \left\{\begin{array}{ll}
e^{\beta (4d x +(2k+1)y+1)} (1-e^{-\beta(y+1)}) & \mbox{if $y\geq 1$}\\
e^{\beta (4d x +(2k+1)y+1)}(1-e^{-2\beta}) & \mbox{if $|y|< 1$}\\
e^{\beta (4d x +(2k+1)y+1)} (1-e^{\beta(y-1)})& \mbox{if $y\leq -1$}
\end{array}\right..
\end{align}

\begin{lema}\label{lemdob16.1} Let $(x,y)\in A\cup B\cup C$,  where $A$, $B$ and $C$  are defined by (\ref{defa}), (\ref{defb}) and (\ref{defc}), respectively. Suppose that  $\sigma\equiv \tilde{\sigma} \mbox{ off $1$}$ and $|\sigma_1|\neq  |\tilde{\sigma}_1|$, then

\begin{align*}
d(\pi^{\{0\}}_{x,y,\beta}(.|\sigma),\pi^{\{0\}}_{x,y,\beta}(.|\tilde{\sigma}))&\leq &
\left\{
\begin{array}{ll}
3 e^{\beta (2dx+2d(y+1))} (1 -e^{-\beta  (y+1) })&\mbox{ if $(x,y)\in A$}\\
3 e^{\beta (2dx+2d(y+1))} (1-e^{-2\beta})& \mbox{ if $(x,y)\in B$}\\
3 e^{2d \beta x}   (1 -e^{\beta  (y-1) })& \mbox{ if $(x,y)\in C$}\\
\end{array}\right..
\end{align*}
\end{lema}

 \begin{proof}  Suppose first that $(x,y)\in A$, then $y\geq 1$. Therefore, from (\ref{eqdob15.2}), (\ref{eqdob16.2w1}) and  (\ref{eqjuly8z}), since $2dx+k(y-1)<x<0$ in $A\cup B\cup C$, $y+1\geq 0$ and  $k+1\leq 2d$, we have
\begin{align*}
\mbox{}&\mbox{} d(\pi^{\{0\}}_{x,y,\beta}(.|\sigma),\pi^{\{0\}}_{x,y,\beta}(.|\tilde{\sigma}))\\
\mbox{}&\leq ( e^{\beta (4d x +(2k+1)y+1)}+2e^{\beta(2dx+(k+1)(y+1))}) (1-e^{-\beta(y+1)})
\\
\mbox{}&=   e^{\beta (2dx+(k+1)(y+1))}(e^{\beta(2dx+k(y-1))}+2 ) (1 -e^{-\beta  (y+1) })
\\
\mbox{}&\leq 3 e^{\beta (2dx+(k+1)(y+1))} (1 -e^{-\beta  (y+1) })\\ &\leq  3 e^{\beta (2dx+2d(y+1))} (1 -e^{-\beta  (y+1) }).\end{align*}
\indent Suppose  that $(x,y)\in B$, then   $|y|<1$. Therefore, from (\ref{eqdob15.2}), (\ref{eqdob16.2w1}) and  (\ref{eqjuly8z}), proceeding exactly as we did in the previous case,  we have
\begin{align*}
d(\pi^{\{0\}}_{x,y,\beta}(.|\sigma),\pi^{\{0\}}_{x,y,\beta}(.|\tilde{\sigma}))&\leq 3 e^{\beta (2dx+2d(y+1))} (1-e^{-2\beta}).
\end{align*}
\indent Suppose that $(x,y)\in C$, then $y<-1$. Therefore, from  (\ref{eqdob15.2}), (\ref{eqdob16.2w1}) and  (\ref{eqjuly8z}),  since  $2dx+k(y-1)+y+1<0$ and $k(y+1)\leq  0$, we have
\begin{align*}
d(\pi^{\{0\}}_{x,y,\beta}(.|\sigma),\pi^{\{0\}}_{x,y,\beta}(.|\tilde{\sigma}))&\leq   (e^{\beta ((4d x +(2k+1)y+1)}+2e^{\beta(2dx+k(y+1))})(1-e^{\beta(y-1)}))
 \\
 \mbox{}&=  e^{\beta(2dx+k(y+1))}   (e^{\beta (2dx+k(y-1)+y+1)}+2)(1 -e^{\beta  (y-1) })
 \\
 \mbox{}&\leq 3 e^{\beta(2dx+k(y+1))}   (1 -e^{\beta  (y-1) })\\
& \leq  3 e^{2d \beta x}   (1 -e^{\beta  (y-1) }),\end{align*}
and this concludes the proof of the lemma.  \end{proof}

\section{Concluding Remarks}\label{seccr}

Even though from Theorem \ref{teo0} for any $(x,y)\in {\mathcal U}$  we can get uniqueness of the Gibbs state for both low and high temperatures, in order to get uniqueness for all temperature in Theorem \ref{teo1} we need to take $x$ sufficiently large, depending on $d$ and $y$.  Besides, we can show that for parameters $(0,y)$, where $y<-1$,  we cannot satisfy the condition (\ref{eqdobc16.2}) for low temperature. This means that there is no hope to reach those values of  parameters $(x,y)$ where the mean-field  calculations of \cite{bib:berker1}  and  the numerical renormalization group calculations  \cite{bib:osorio,bib:branco} predictions are  in disagreement, using Dobrushin (or even Dobrushin-Shlosman) criterion. In any case,  Theorem \ref{teo1} rules out the possibility of a reentrant behavior for $(x,y)\in{\mathcal U}^d_{Dob}$.

\end{document}